\documentclass[11pt,a4paper]{article}
\usepackage[utf8]{inputenc}
\usepackage[english]{babel}
\usepackage{amsmath}
\usepackage[amsmath,thmmarks]{ntheorem}
\usepackage{amsfonts}
\usepackage{amssymb}
\usepackage{graphicx}
\usepackage[left=2cm,right=2cm,top=2cm,bottom=2cm]{geometry}

\usepackage{authblk}
\usepackage{varioref}
\usepackage{mathtools}
\usepackage{thm-restate}
\usepackage[linkcolor=black,colorlinks=true,citecolor=black,filecolor=black]{hyperref}
\usepackage[capitalise,nameinlink]{cleveref}
\usepackage{sectsty}
\usepackage{parskip}
\partfont{\centering\Large\selectfont}

\numberwithin{equation}{section}

\theoremstyle{plain}

\newtheorem{theorem}{Theorem}[section]

\newtheorem{remark}[theorem]{Remark}
\newtheorem{corollary}[theorem]{Corollary}
\newtheorem{definition}[theorem]{Definition}
\newtheorem{lemma}[theorem]{Lemma}
\newtheorem{proposition}[theorem]{Proposition}
\newtheorem{problem}[theorem]{Problem}

\newtheorem*{theorem-nn}{Theorem}[section]
\newtheorem*{example-nn}[theorem]{Example}
\newtheorem*{remark-nn}[theorem]{Remark}
\newtheorem*{corollary-nn}[theorem]{Corollary}
\newtheorem*{definition-nn}[theorem]{Definition}
\newtheorem*{lemma-nn}[theorem]{Lemma}
\newtheorem*{proposition-nn}[theorem]{Proposition}
\newtheorem*{problem-nn}[theorem]{Problem}

\theoremstyle{nonumberplain}
\theorembodyfont{\normalfont}
\theoremsymbol{\ensuremath{\square}}
\newtheorem{proof}{Proof}

\theoremstyle{plain}

\newcommand{\R}{\mathbb{R}}
\newcommand{\Z}{\mathbb{Z}}

\newcommand{\F}{\mathcal{F}}
\newcommand{\E}{\mathcal{E}}
\newcommand{\SIt}{\mathcal{S}}
\newcommand{\bigO}{\mathcal{O}}

\DeclarePairedDelimiter{\floor}{\lfloor}{\rfloor}

\renewcommand{\epsilon}{\ensuremath\varepsilon}

\renewcommand{\phi}{\ensuremath{\varphi}}

\title{Efficiently stabbing convex polygons and variants of the Hadwiger-Debrunner $(p, q)$-theorem}
\author[1]{Justin Dallant \thanks{This work was partially supported by the French Community of Belgium via the funding of a FRIA grant.}}
\author[2]{Patrick Schnider \thanks{Has received funding from the European Research Council under the European Unions Seventh Framework Programme ERC Grant agreement ERC StG 716424 - CASe. Part of this work was done when the author was employed at ETH Zürich.}}
\affil[1]{Department of Computer Science, Université Libre de Bruxelles\\
  \texttt{justin.dallant@ulb.be}}
\affil[2]{Department of Mathematical Sciences, University of Copenhagen\\
  \texttt{ps@math.ku.dk}}

\begin{document}

\maketitle

\begin{abstract}
Hadwiger and Debrunner showed that for families of convex sets in $\R^d$ with the property that among any $p$ of them some $q$ have a common point, the whole family can be stabbed with $p-q+1$ points if $p \geq q \geq d+1$ and $(d-1)p < d(q-1)$. This generalizes a classical result by Helly. We show how such a stabbing set can be computed for a family of convex polygons in the plane with a total of $n$ vertices in  $\bigO((p-q+1)n^{4/3}\log^{8} n(\log\log n)^{1/3} + np^2)$ expected time. For polyhedra in $\R^3$, we get an algorithm running in $\bigO((p-q+1)n^{5/2}\log^{10} n(\log\log n)^{1/6} + np^3)$ expected time. We also investigate other conditions on convex polygons for which our algorithm can find a fixed number of points stabbing them.
Finally, we show that analogous results of the Hadwiger and Debrunner $(p,q)$-theorem hold in other settings, such as convex sets in $\R^d\times\Z^k$ or abstract convex geometries.
\end{abstract}

\part*{Introduction}

A classical result in convex geometry by Helly \cite{Helly1923} states that if a family of convex sets in $\R^d$ is such that any $d+1$ sets have a common intersection, then all sets do.
In 1957, Hadwiger and Debrunner \cite{Hadwiger1957} considered a generalization of this setting. Let $\F$ be a family of sets in $\R^d$ and let $p\geq q \geq d+1$ be integers. We say that $\F$ has the $(p,q)$-property if $|\F|\geq p$ and for every choice of $p$ sets in $\F$ there exist $q$ among them which have a common intersection. We further say that a set of points $S$ stabs $\F$ if every set in $\F$ contains at least one point from $S$. Then the following holds.

\begin{restatable}{theorem}{thmHD}[{{Hadwiger and Debrunner \cite{Hadwiger1957}}}]
\label{thm:HD}
Let $d\geq 1$ be an integer. Let $p$ and $q$ be integers such that $p\geq q\geq d+1$ and $(d-1)p<d(q-1)$, and let $\F$ be a finite family of convex sets in $\R^d$. Suppose that $\F$ has the $(p,q)$-property. Then there exist $p-q+1$ points in $\R^d$ stabbing $\F$.
\end{restatable}

Note that the bound on the number of points needed is tight. That is, for every $p\geq q\geq d+1$ there exist families of convex sets with the $(p,q)$-property where at least $p-q+1$ points are needed to stab the whole family. This is easily seen by considering any family of $p-q+1$ disjoint convex sets where one of them is taken with multiplicity $q$. It is also known that whenever $q \leq d$, there exist families of convex sets with the $(p,q)$-property where arbitrary large number of points are needed. This can be seen by taking $n$ hyperplanes in general position in $\R^d$ (meaning that no two hyperplanes are parallel and no $d+1$ hyperplanes intersect at the same point). Then any $d$ hyperplanes intersect at some point (in other words, they have the $(d,d)$-property) and any single point stabs at most $d$ hyperplanes. Thus, at least $\floor{n/d}$ points are necessary to stab all hyperplanes.

Many related results have since been established. Among the most famous is one from Alon and Kleitman \cite{Alon1992} who in 1992 proved that for any $p\geq q \geq d+1$, there exists a finite upper bound on the maximum number of points needed to stab a family of convex sets with the $(p,q)$-property. However, all the known upper bounds are probably far from being tight in the general case. As an example, for $(p,q,d) = (4,3,2)$, their proof yields an upper bound of $4032$ (while the best known lower bound is $3$). Since then, this number has been proven to lie between $3$ and $13$ (inclusive) \cite{Kleitman2001}. Still, the only values of $p\geq q \geq d+1$ for which exact values are known are those corresponding to \cref{thm:HD}. There is a lot of work in this more general setting, both improving the bounds (e.g. \cite{Keller2018}) as well as adapting to generalizations of convex sets (e.g. \cite{Holmsen2019, Moran2019}), and it is an interesting open problem to study algorithmic questions connected to these results.

Special cases where some further restrictions are imposed on the considered sets have also led to interesting results. One much studied example is obtained by considering only axis-aligned boxes in $\R^d$. In this case, we can already start by strengthening the result given by Helly's theorem, as for a family of axis-aligned boxes in $\R^d$, if all pairs intersect then the whole family intersect. As is expected, this additional structure leads to stronger $(p,q)$ results.  One early result appearing in \cite{Hadwiger1964} is the following (notice the weaker conditions on $p$ and $q$ and the independence on $d$).

\begin{theorem}[\cite{Hadwiger1964}]
Let $d\geq 1$ be an integer. Let $p$ and $q$ be integers such that $2q-2\geq p \geq q \geq 2$ and let $\F$ be a finite family of axis-aligned boxes in $\R^d$. Suppose that $\F$ has the $(p,q)$-property. Then there exist $p-q+1$ points in $\R^d$ stabbing $\F$.
\end{theorem}

Another example is when all sets are translations either with or without scaling of some convex set $K$. Here, strong results exist only for some very simple cases such as $K$ being a $d$-dimensional cube or ball. For example the maximum number of points needed to stab families of discs in the plane with the $(p,2)$-property lies between $4p-4$ and $7p-10$ inclusive \cite{Wegner1985}. These bounds are tight for $p=2$, that is, in the case of pairwise intersecting discs.

From an algorithmic point of view, little work seems to have been done towards computing these stabbing points. One instance which has recently received some attention is the aforementioned case of pairwise intersecting discs in the plane. In \cite{Harpeled2018}, it was shown how such a family can be stabbed with $5$ points in linear time (which is one more point than the theoretical optimum). Shortly after a linear time algorithm for stabbing such a family with only $4$ points was found \cite{Carmi2018}. However, the computation of small stabbing sets for families of general convex polyhedra with the $(p,q)$-property seems to not have been studied and will constitute one part of this paper, in the setting of \cref{thm:HD}.

For a great overview of the studied questions and known results around $(p,q)$ problems, we refer the interested reader to the 2003 survey by Eckhoff \cite{Eckhoff2003}.

Before continuing, we would also like to mention that Helly's theorem has been generalized to many other settings, as this will come in play in the second part of this paper. In general, we say that a set system has \emph{Helly number} $h$ if the following holds: if any $h$ sets in the set system have a common intersection, then the whole set system does.
Helly numbers have been shown to exist for many set systems, such as convex sets in $\R^d\times\Z^k$ \cite{Averkov2012, Hoffman2006} or abstract convex geometries (see \cite{Edelman1985} or Chapter III of \cite{Korte1991}), which include subtrees of trees and ideals of posets.
In many of these cases, the proofs can be adapted to show a constant stabbing number analogous to the result by Alon and Kleitman.
In this work, we will show that under some weak conditions, the existence of a Helly number implies a tight Hadwiger-Debrunner type result.

\part*{Stabbing convex polytopes}

\section{The general dimension case}

\subsection{A proof of the Hadwiger-Debrunner theorem}

We will first consider a proof of \cref{thm:HD} which will naturally lead to an algorithm for finding stabbing points. In \cite{Matousek2002}, the proof of this theorem is left as an exercise, yet the hint suggests that the intended solution is close to the proof below. The main differences with other proofs for this theorem are that it is more constructive and does not make use of a separating hyperplane, which will make it easier to adapt to other settings later on.

We will make use of a lemma which can also be found in \cite{Matousek2002}. We include the proof as we will later use the same ideas to prove a similar lemma. For a non-empty compact set $S$, let $lexmin(S)$ denote its lexicographical minimum point. Then we have the following.

\begin{lemma}\label{lemma:lexmin}
Let $\F$ be family of at least $d+1$ compact convex sets in $\R^d$, such that $I := \bigcap \F$ is non-empty. Let $x := lexmin(I)$. Then, there exist a subfamily $\mathcal{H} \subset \F$ of size $d$ such that $x = lexmin(\bigcap \mathcal{H})$.
\end{lemma}
\begin{proof}
Let $\F$, $I$ and $x$ be as specified in the statement. Let $S_x$  denote the set of all points lexicographically smaller than $x$. This set is convex and is disjoint from $I$. By Helly's theorem, there exists a subfamily of $d+1$  members of $\F\cup\{S_x\}$ with an empty common intersection. These members have to include $S_x$, as all members of $\F$ have a non-empty common intersection. Let $\mathcal{H} \subset \F$ be the family consisting of the remaining $d$ sets and let $x_\mathcal{H}$ be the lexicographical minimum point of $I':=\bigcap \mathcal{H}$ (which is compact and non-empty). $x_\mathcal{H}$ can not be lexicographically larger than $x$ because $\mathcal{H} \subset \F$ and it can not be lexicographically smaller than $x$ because $I'\cap S_x = \emptyset$. Thus, $x_\mathcal{H} = x$.
\end{proof}

Recall the theorem we wish to prove:

\thmHD*

\begin{proof}
We will prove the theorem for families of compact convex sets, as we will only deal with such families later. One can however reduce the original theorem to this one (see \cref{app:compact}), so this is done without loss of generality.

Call a pair of integers $(p, q)$ \emph{admissible} if $p\geq q \geq d+1$ and $(d-1)p<d(q-1)$. Let $(p,q)$ be an admissible pair, and let $\F$ be a family of compact convex sets of $\R^d$ with the $(p,q)$-property. 

We reason by induction on $p$, the base case being $p=q=d+1$ which is Helly's theorem.

If $p=q>d+1$, then $\F$ also has the $(d+1,d+1)$ property (as having the $(p,q)$-property implies having the $(p-1,q-1)$-property) and the result again follows from Helly's theorem.

So suppose that $p>q$ and that the result is true for any admissible pair $(p',q')$ with $p'<p$. 

If $(d-1)p=d(q-1)-k-1$ for $k\geq 1$, then notice that $(p-k,q-k)$ is an admissible pair, as in that case $(d-1)(p-k) = d(q-k-1)-1$ which together with $p>q$ also implies that $q-k \geq d+1$. Thus the result follows from the induction hypothesis.  

It now remains to consider the case where $p>q$ and $(d-1)p = d(q-1)-1$. 

To do so, let us construct a point $x^*(\F)$ as follows: 
\begin{itemize}
\item For every non-empty subfamily $\SIt \subset \F$ of $d$ convex sets with non-empty intersection, let $x_{\SIt}$ be the lexicographical minimum of $I_{\SIt} = \bigcap \SIt$.
\item Let $x^*(\F)$ be the lexicographical maximum point among all such $x_{\SIt}$'s.
\end{itemize}

Let $\mathcal{G}$ be one of the families defining $x^*(\F)$, that is, $\mathcal{G} \subset \F'$ is a subfamily of $d$ sets which have $x^*(\F)$ as the lexicographical minimum of their intersection. 

To establish the theorem, it is enough to show that by choosing $x^*(\F)$ as one of our stabbing points, we can stab all the remaining sets (i.e. those which do not contain $x^*(\F)$) with $p-q$ points. Let $\mathcal{R} = \{C \in \F\ |\ x^*(\F) \not\in C\}$ be the set of remaining sets.

Let us argue that for any $S\in \mathcal{R}$, $S \cap (\bigcap \mathcal{G})$ is empty. To do so, suppose it was not, and let $y$ be the lexicographical minimum of that intersection. By \cref{lemma:lexmin}, $y$ is the lexicographical minimum of the intersection of $d$ sets in $\F$. Moreover, by definition of $\mathcal{R}$ and $\mathcal{G}$, $y$ is lexicographically larger than $x^*(\F)$. This contradicts the definition of $x^*(\F)$. Thus, $S \cap (\bigcap\mathcal{G})$ is empty.

Two cases arise:
\begin{enumerate}
\item $(|\mathcal{R}| \geq p-d)$ We show that $\mathcal{R}$ has the $(p-d,q-d+1)$-property. Indeed, choose any $p-d$ members from $\mathcal{R}$ together with the $d$ members from $\mathcal{G}$. We know from the $(p,q)$-property of $\F$ that there exists a subfamily $\mathcal{E} \subset \mathcal{R}\cup \mathcal{G}$ of size $q$ whose members have a non-empty common intersection. $\mathcal{E}$ cannot contain all elements of $\mathcal{G}$, as $q>d = |\mathcal{G}|$ and the intersection of all members of $\mathcal{G}$ together with any member of $\mathcal{R}$ is empty. Thus, $\mathcal{E}$ contains at least $q-d+1$ members of $\mathcal{R}$. This shows that $\mathcal{R}$ has the $(p-d,q-d+1)$-property. Notice that with the assumptions $p>q$ and $(d-1)p = d(q-1)-1$ which we are working under, $(p-d, q-d+1)$ is admissible. Thus, by the induction hypothesis, $\mathcal{R}$ can be stabbed with $p-d - (q-d+1) + 1 = p-q$ points.
\item $(|\mathcal{R}| < p-d)$ In this case, choose $\mathcal{R}$ as a whole together with $\mathcal{G}$ and $p-d-|\mathcal{R}|$ other members of $\F$. By the same reasoning as in case 1., there exists a subset of $q-(d-1+p-d-|\mathcal{R}|) = |\mathcal{R}|+ 1 + q - p$ members of $\mathcal{R}$ which intersect and can thus be stabbed by a single point. The remaining $|\mathcal{R}| - (|\mathcal{R}|+ 1 + q - p) = p-q-1$ sets can trivially be stabbed by $p-q-1$ points.
\end{enumerate}

Thus, $\mathcal{R}$ can be stabbed by $p-q$ points, which implies that $\F$ can be stabbed by $p-q+1$ points and by induction, concludes the proof.
\end{proof}

\subsection{A first algorithm}

This proof naturally leads to an algorithm. Let $d>0$ be some fixed dimension and let $\F$ be a family of compact convex polytopes with $(p,q)$-property, described as intersections of a total of $n$ halfspaces in general position. For simplicity we assume that the common intersection of any $d$ of these polytopes is either empty or contains a unique point with minimum $x$-coordinate (which is then also the lexicographically minimum point in the intersection).

The algorithm works as follows:

\begin{enumerate}
\item Reduce $p$ and $q$ (as done in the proof of \cref{thm:HD}) to reach the case where $p=q=d+1$ or the case where $p>q$ and $(d-1)p = d(q-1)-1$. 
\item Construct a point $x^*(\F)$ defined as in the proof. We choose it as one of our stabbing points.
Now, remove from $\F$ all the sets that are stabbed by this point. If there are any remaining sets then either $|\F| \geq p-d$ and $\F$ satisfies the $(p-d,q-d+1)$-property, where $(p-d,q-d+1)$ is admissible, or $\F$ consists of $p-q+k$ sets, $k < q-d$, where some $k+1$ of them have a common intersection.
\item In the first case, we can continue inductively.
\item In the second case we can trivially stab the remaining sets using $p-q$ points.
\end{enumerate}

The correctness of the algorithm follows immediately from the proof of \cref{thm:HD}. The only detail that needs some additional scrutiny is the correctness for the base case $p=q=d+1$. Notice that in this case all sets have a common intersection and \cref{lemma:lexmin} ensures that $x^*(\F)$ stabs the whole family $\F$.

Regarding the runtime of Step 2, the most natural way to compute $x^*(\F)$ gives the following.
\begin{lemma}\label{lemma:d-wise-intersect}
We can compute $x^*(\F)$ in $\bigO(n|\F|^{d-1})$ time.
\end{lemma}
\begin{proof}
For every polytope $P\in \F$, we let $n(\mathcal{G})$ denote the total number of halfspaces describing $P$.
For every subfamily $\mathcal{G}$ of $\F$, we let $n(\mathcal{G}) := \sum_{P\in\mathcal{G}}n(P)$.

We can compute $x^*(\F)$ by testing for intersection in every subfamily $\mathcal{G}$ of $\F$ of size $d$ and computing the point with minimum $x$-coordinate of that intersection if it is non-empty. We then take the lexicographically maximum point among all those computed.

If we consider some fixed subfamily $\mathcal{G}$, this computation can be done in $\bigO(n(\mathcal{G}))$ time using linear programming in constant dimension. Thus, the computation for that subfamily will cost at most $c\cdot n(\mathcal{G})$ for some constant $c$ which does not depend on $\mathcal{G}$. Charge this cost to the polytopes $P \in \mathcal{G}$ by attributing a cost of  $c\cdot n(P)$ to a polytope $P$. 

Now, consider the cost charged to some fixed polytope $P$ for the whole computation. As $P$ appears in no more than $|\F|^{d-1}$ subfamilies of size $d$, its total cost charge is upper bounded by $c\cdot n(P)\cdot |\F|^{d-2}$. Summing across all polytopes $P\in\F$, we get a total cost of $\bigO(n\cdot|\F|^{d-1}))$.
\end{proof}

This quantity needs to be computed at most $p-q+1$ times with the family $\F$ decreasing in size each time.

For Step 4, we have the following.
\begin{lemma}\label{lemma:base-stab-poly}
For any $k>0$, we can find a point stabbing $k$ polytopes of $\F$ in $\bigO(n|\F|^d)$ time, if such a point exists.
\end{lemma}
\begin{proof}
If $k\leq d$, then we can test every subfamily of size $k+1$ for common intersection and compute a point in the intersection for a total cost of $\bigO( n|\F|^{k}) \leq \bigO( n|\F|^{d})$.

If $k> h-1$, then we know from \cref{lemma:lexmin} that the lexicographical minimum of the intersection of $k$ convex polytopes is also the the lexicographical minimum of the intersection of some $d$ polytopes in $\F$ (which in our case is also the point with minimum $x$-coordinate in the intersection). Thus, one can find a point stabbing at least $k$ sets by computing the point with minimum $x$-coordinate for each subfamily of size $d$ (in $\bigO(n|\F|^{d})$ time) and counting the number of sets intersected for each of the $\bigO(|\F|^{d})$ computed points (in $\bigO(d|\F|^{d})$ time as well).
\end{proof}

Because we know that when reaching Step 4 we have $|\F| < p$, it follows that Step 4 can be done in $\bigO(np^d)$ time.

Thus, we get a total runtime of 
\[\bigO((p-q+1)n^d + np^{d}).\]

If $p$ (and thus $q$) is small compared to $n$, the bottleneck in the computation time the first term, which scales as $\bigO(n^d)$ with respect to $n$. The natural question that now comes to mind is: can we do better than $\bigO(n^d)$? We will see in the following section that we can indeed do better at least in dimensions 2 and 3, although at the cost of considering expected rather than worst-case runtime.

\begin{remark} \label{remark:constant_poly}
If we further restrict the problem to only consider convex polytopes described by at most a constant number of halfspaces each, then the second term in the runtime becomes $\bigO(p^{d+1})$. In the plane, this term can further be improved from $\bigO(p^3)$ to $\bigO(p^2\log p)$ by adapting the Bentley-Ottmann sweep line algorithm \cite{Bentley1979} (see \cref{app:bentley-ottmann} for more details). On the other hand, one can easily reduce the problem of finding a point stabbing at least three lines among $p$ lines to the problem of Step 4 in the above algorithm (for $k>1$) in linear time if we allow for infinitesimally thin polygons. This problem is 3-SUM hard (see \cite{Gajentaan1995}, where the concept of 3-SUM hardness was first introduced). There is a strong belief that such problems can not be solved in $O(p^{2-\epsilon})$ time, which means that Step 4 can probably not be solved in $O(p^{2-\epsilon})$ time either, even for constant-size polygons.
\end{remark}

\section{Faster algorithms for 2D and 3D polytopes}

In what follows we deal with the cases $d=2$ and $d=3$. Note that for $d\leq 3$ we can get the vertex representation of our polytopes as well as the faces of all dimension from the halfspace representation in $O(n\log n)$ time by computing the convex hulls of the dual point sets. Thus we will assume that we have access to the vertices and edges and faces of our polygons and polyhedra as the $O(n\log n)$ overhead will be dominated by the rest of our algorithms.

\subsection{The planar case}

In this whole section, the family $\F$ consists of compact convex polygons with a total of $n$ (distinct) vertices in the plane and has the $(p,q)$-property, for some admissible pair $(p,q)$. For the sake of simplicity, we will assume that the lines defining the polygon edges are in general position, non-vertical and that all points defined as the lexicographical minimum in the intersection of a pair of sets have different $x$-coordinates. Under these assumption the lexicographical minimum in a polygon (or intersection of polygons) is simply the leftmost point. 

We break down the computation of $x^*(\F)$ into two parts. Consider two intersecting polygons $P_1$ and $P_2$. The point $x$ which is the leftmost of $P_1\cap P_2$ can be of one of two types. Either (case 1) $x$ is the leftmost point of $P_1$ (resp. $P_2$) and is contained in the interior of $P_2$ (resp. $P_1$) or (case 2) $x$ is the proper intersection of an upper-hull edge $e_u$ of $P_1$ (resp. $P_2$) and a lower-hull edge $e_\ell$ of $P_2$ (resp. $P_1$) with the following property: the outward facing normal vectors of $e_U$ and $e_L$ form a (counter-clockwise orientated) angle of less than $180$ degrees. Reciprocally, an upper-hull and a lower-hull edge which intersect with this property define the leftmost point of an intersection of two polygons. 

We define $x_1^*(\F)$ to be the rightmost point among all pairs of intersecting polygons in $\F$ corresponding to the first case (or $x_1^*(\F)=(-\infty, \infty)$ if there is no such pair), and similarly for $x_2^*(\F)$ and the second case. It is clear that $x^*(\F)$ is the rightmost point of $\{x_1^*(\F), x_2^*(\F)\}$.

We will use the following result, which can be obtained by an adaptation of the proof of Matou\v{s}ek's Theorem 6.2 in \cite{Matousek1992} with the halfspace partition tree construction from Chan \cite{Chan2016} (see \cref{app:range-query}).
\begin{restatable}{theorem}{thmRangeQuery}\label{thm:range-query}
Let $S$ be a set of $n$ objects, $k$ a constant, and $\phi_1, \phi_2, \ldots, \phi_k$ mappings from $S$ to $\R^d$. Let $\phi_S$ be the function which maps $k$-tuples of halfspaces $H_1, H_2, \ldots, H_k$ of $\R^d$ to the set
\[\phi_S(H_1, H_2, \ldots, H_k) :=\{s\in S \mid \phi_1(s) \in H_1,\phi_2(s) \in H_2\ldots, \phi_k(s) \in H_k \}.\]
Suppose we have computed the point sets $\phi_1(S), \ldots, \phi_k(S)$ and let $n \leq m \leq n/\log^{\omega(1)}n$. Then we can preprocess the point sets in $\bigO(n\log^{k} n + m)$ time such that $|\phi_S(H_1, H_2, \ldots, H_k)|$ can be computed in $\bigO((n/m^{1/d})(\log n)^{2(k+(k-d-1)/d)}(\log\log n)^{1/d})$ expected time for any $k$-tuple of halfspaces.
\end{restatable}
Note that we have made no big effort in minimizing the polylog factor in the query runtime. It is thus conceivable that a more careful use of the tools in \cite{Matousek1992, Chan2016} could make this factor smaller.

We can use this result to prove the following.
\begin{lemma}
We can compute $x_1^*(\F)$ in $\bigO(n^{4/3}\log^{4} n(\log\log n)^{1/3})$ expected time.
\end{lemma}
\begin{proof}
To compute $x_1^*(\F)$, we can test for each polygon if its leftmost point is contained in the interior or another, and keep the rightmost point among those which are. We triangulate all polygons, so that this reduces to testing, for each of the $\bigO(n)$ leftmost points, if it is in the interior of one of the $\bigO(n)$ triangles. In the dual plane, this can be expressed as the composition of three half-plane range queries. Using \cref{thm:range-query} with $d=2$, $k=3$, and $m = n^{4/3}\log^{4} n(\log\log n)^{1/3}$, we can thus preprocess the $\bigO(n)$ triangles in $\bigO(m)$ time such that counting how many triangles contain a particular point can be done in  $\bigO(n^{1/3}\log^{4} n(\log\log n)^{1/3})$ expected time. By querying all points we get the result.
\end{proof}

It remains to see how to compute $x_2^*(\F)$ in subquadratic time. For this we use a a simple but remarkably powerful technique discovered by Chan \cite{Chan1999}, which reduces many optimization problems to the corresponding decision problem, with no blow-up in expected runtime.

\begin{lemma}
\label{lemma:chan}
Let $\alpha<1$ and $r$ be fixed constants. Let $f: \mathcal{P} \rightarrow \mathcal{Q}$ be a function that maps inputs to values in a totally ordered set (where elements can be compared in constant time) with the following properties.
\begin{enumerate}
\item For any input $P\in \mathcal{P}$ of constant size, $f(P)$ can be computed in constant time.
\item For any input $P\in \mathcal{P}$ of size $n$ and any $t\in \mathcal{Q}$, we can decide $f(P)\leq t$ in time $T(n)$.
\item For any input $P\in \mathcal{P}$ of size $n$, we can construct inputs $P_1,\ldots,P_r \in \mathcal{P}$ each of size at most $\lceil \alpha n \rceil$ in time $T(n)$, such that 
$f(P) = \max\{f(P_1),\ldots,f(P_r)\}$.
\end{enumerate}
Then for any input $P\in\mathcal{P}$, we can compute $f(P)$ in $\bigO(T(n))$ expected time, assuming that $T(n)/n^\epsilon$ is monotone increasing for some constant $\epsilon > 0$.
\end{lemma}

We can apply this technique to the computation of $x_2^*$.
Here, each $P\in\mathcal{P}$ is a set of edges which are oriented depending on which side the polygon it bounds lies on, $\mathcal{Q}$ is the plane with lexicographical order, and $f(P)$ is $x_2^*$ (we abuse notation slightly by using $x_2^*$ both for sets of oriented edges and sets of polygons).
We make the following observations.

\begin{enumerate}
\item For any constant-size set $\E$ of oriented edges, $x_2^*(\E)$ can be computed in constant time. This verifies property 1.
\item For any family $\E$ of $n$ oriented edges, we can partition it into $3$ disjoint subfamilies $S_1,S_2,S_3$ of size between $\lfloor n/3 \rfloor$ and $\lceil n/3 \rceil$ each. Then, let $\E_1 := S_2\cup S_3$, $\E_2 := S_1\cup S_3$ and $\E_3 := S_1\cup S_2$. Every set $\E_i$ is of size $|\E_i| \leq \lceil 2n/3 \rceil$. Thus, $x_2^*(\F)$ is the rightmost point among $\{x_2^*(\F_1),x_2^*(\F_2), x_2^*(\F_3)\}$. These families can be constructed in $\bigO(n)$ time. This verifies property 3, assuming $T(n) \geq \Omega(n)$ (which it will be).
\end{enumerate}

Thus, in order to apply Chan's framework, it remains to decide $x_2^*(\E)\leq_{lex} t$ quickly.
\begin{lemma}
For any point $t$ in the plane and a set of $n$ oriented edges $\E$, we can decide $x_2^*(\E)\leq_{lex} t$ in $\bigO(n^{4/3}\log^{8} n(\log\log n)^{1/3})$ expected time.
\end{lemma}
\begin{proof}
We can rephrase $x_2^*(\E)\leq_{lex} t$ as deciding whether there exist two oriented edges in $\E$ which intersect at an appropriate angle to the right of the vertical line $\ell$ passing through $t$. Thus we start by discarding all the (parts of) segments in $\E$ which lie to the left of $\ell$.
We then want to preprocess the $\bigO(n)$ segments corresponding to upper-hull edges (i.e. those with an outward facing normal pointing up) such that for any lower-hull edge $e_L$ we can detect if there is an upper-hull edge which intersects it at an appropriate angle quickly.

Map each upper-hull edge $e_U$ to its endpoints $a(e_U)$, $b(e_U)$ and to the point $p^*(e_U)$ dual to the line supporting it. Now for a lower-hull edge $e_L$, let $R(e_L)$ denote the region of the plane corresponding to all points whose dual line intersects $e_L$ an appropriate angle. This region is a convex polygon with at most $4$ edges. Thus it can be partitioned into two triangles $R_1(e_L)$ and $R_2(e_L)$. Call $\ell$ the line supporting $e_L$. Now, all upper-hull edges $e_U$ intersecting $e_L$ at an appropriate angle fall into exactly one of these categories:
\begin{itemize}
    \item $a(e_U)$ lies to the left of $\ell$, $b(e_U)$ lies to the right of $\ell$ and $p^*(e_U)\in R_1(e_L)$,
    \item $a(e_U)$ lies to the left of $\ell$, $b(e_U)$ lies to the right of $\ell$ and $p^*(e_U)\in R_2(e_L)$,
    \item $a(e_U)$ lies to the right of $\ell$, $b(e_U)$ lies to the left of $\ell$ and $p^*(e_U)\in R_1(e_L)$,
    \item or $a(e_U)$ lies to the right of $\ell$, $b(e_U)$ lies to the left of $\ell$ and $p^*(e_U)\in R_2(e_L)$.
\end{itemize}

The number of upper-hull edges corresponding to each category can be counted by a range query which is the composition of $5$ half-plane queries on the $3$ defined liftings.

We can again use \cref{thm:range-query} as we did for $x_1^*$, this time with $k=5$, to query all lower-hull edges in $\bigO(n^{4/3}\log^{8} n(\log\log n)^{1/3})$ expected total time.
\end{proof}
We can thus use \cref{lemma:chan} to compute $x_2^*(\F)$ in the same asymptotic expected time. Note that Hopcroft's problem reduces to computing $x_2^*(\E)$ for a general set of oriented edges $\E$, and thus this runtime is likely close to optimal (see \cite{Erickson1996} for a lower bound in a quite general model of computation).

Putting everything together we get the following.
\begin{theorem} \label{thm:stab2D}
Let $(p,q)$ be an admissible pair for $d=2$ and let $\F$ be a family of compact convex polygons in the plane with a total of $n$ vertices and the $(p,q)$-property. Then we can compute a set of at most $p-q+1$ points stabbing $\F$ in $\bigO((p-q+1)n^{4/3}\log^{8} n(\log\log n)^{1/3} + np^2)$ expected time.
\end{theorem}

\subsection{Constant-size polygons}
If we restrict all polygons to have at most a constant number of vertices, then a simpler proof using Theorem 2.7 in \cite{AGARWAL2002} yields a slightly faster algorithm. Indeed, this theorem states the following.
\begin{theorem}
Let $\F$ be a family of compact convex polygons in the plane with a total of $n$ vertices. Then, we can count the number of pairs of polygons in $\F$ which intersect in $\bigO(n^{4/3}\log^{2+\epsilon} n)$ time, for any constant $\epsilon > 0$.
\end{theorem}

Using this, it is easy to prove the following.
\begin{theorem} \label{thm:right_intersection}
Given a vertical line $\ell$  and a family $\F$ of compact convex polygons in the plane with a total of $n$ vertices, we can decide whether $x^*(\F)$ lies to the right of $\ell$ in  $O(n^{4/3}\log^{2+\epsilon} n)$ time, for any constant $\epsilon > 0$.
\end{theorem}
\begin{proof}
We start by cutting all polygons along the vertical line $\ell$ and discarding the parts lying on the left of $\ell$ in linear time. 

The point $x^*(\F)$ lies to the right of $\ell$ if and only if there are two polygons which have a non-empty intersection but do not intersect on $\ell$. This can be decided by counting the number of pairwise intersecting polygons in $\bigO(n^{4/3}\log^{2+\epsilon} n)$ time, counting the number of pairwise intersecting polygons on $\ell$ (this can be done in $\bigO(n\log(n))$ time, see \cref{app:intersect-segment}, or one can use the same algorithm again), and then comparing these numbers. They differ if and only if some pair of polygons intersect exclusively to the right of $\ell$.

This whole procedure leads to an algorithm with a $\bigO(n^{4/3}\log^{2+\epsilon} n)$ runtime.
\end{proof}

Together with the modification mentioned in \cref{remark:constant_poly} and a straightforward application of \cref{lemma:chan}, this yields the following.
\begin{theorem}
Let $(p,q)$ be an admissible pair for $d=2$ and let $\F$ be a family $n$ of compact convex polygons in the plane with at most a constant number of vertices each and the $(p,q)$-property. Then we can compute a set of at most $p-q+1$ points stabbing $\F$ in $\bigO((p-q+1)n^{4/3}\log^{2+\epsilon}n + p^2\log p)$ expected time, for any constant $\epsilon > 0$.
\end{theorem}

\begin{remark}
At first glance it might seem like the constant-size assumption plays no essential role here for the $n^{4/3}\log^{2+\epsilon}n$ term in the runtime, and that nothing stops us from using the same approach in the general case. The trouble in the general case however comes in enforcing point 3 in \cref{lemma:chan}. When the polygons are not restricted in size, it might be impossible to create subproblems of appropriate size. This is not a problem in our proof as we are working with sets of oriented edges instead of sets of polygons.
\end{remark}

We finish this part by proving a lower bound on a restricted case of the problem solved in \cref{thm:right_intersection} (thus, a lower bound on the general case also).

\begin{theorem}
Given a vertical line $\ell$  and a family $\F$ of $n$ closed triangles in the plane which all intersect $\ell$, detecting whether two triangles intersect exclusively to the right of $\ell$ requires $\Omega(n\log n)$ time for the worst case in the algebraic decision tree model.
\end{theorem}

\begin{proof}
We prove the claim by a reduction from the Element Uniqueness Problem, which is known to have $\Theta(n\log n)$ time complexity in this model \cite{BenOr1983}. The Element Uniqueness Problem is the following: given an array of $n$ integers, test if they are all disctinct.

Let $a$ be an array of length $n$ representing an instance of the Element Distinctness Problem. Construct an instance of the problem we are interested with in $O(n)$ time in the following way.
\begin{itemize}
\item Let $\ell$ be the $y$-axis.
\item For every $k \in \{0,\ldots,n-1\}$, create a triangle with vertex coordinates $(0,n\cdot a[k]+k)$, $(1, 2\cdot a[k])$ and $(1, 2\cdot a[k] + 1)$.
\end{itemize}

\begin{figure}
\centering
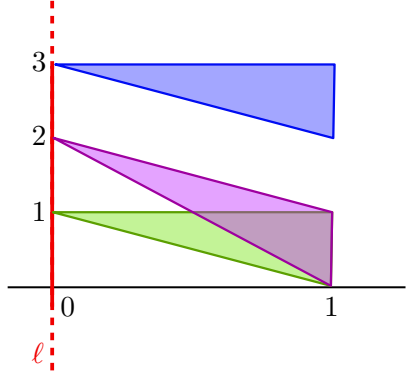
\caption{Right Intersection instance corresponding to the Element Uniqueness instance $[1,0,0]$}
\label{fig:reduction}
\end{figure}

All triangles trivially intersect the vertical line $\ell$ as they have a vertex lying on the $y$-axis.

Let $k,k' \in \{0,\ldots,n-1\}$, such that $k\neq k'$.

Suppose that $a[k] = a[k']$. Because the corresponding triangles share two vertices with each other they have a non-empty intersection. Moreover, this intersection lies entirely to the right of the $y$-axis, as the respective leftmost vertices of the triangles lie on this line and are distinct. 

Now suppose that $a[k] < a[k']$. Then it is easy to see that both triangles lie strictly on a different side of the line passing through the points with coordinates $(0, a[k]+n+1/2)$ and $(1,2\cdot a[k] + 3/2)$, and are thus disjoint.

Thus, one gets a positive response to the Element Distinctness instance if and only if one gets a positive response to this constructed instance. Coupled with the $O(n)$ runtime of the reduction, this concludes the proof. 
\end{proof}

The relevance of the requirement that all polygons intersect $\ell$ is that in the case of polygons with at most a constant number of vertices and having the $(p,q)$ property, we can reduce the problem in \cref{thm:right_intersection} to this case in $\bigO(np)$ time.

\subsection{The 3D case}
Here we deal with the analogous case for the 3D polyhedra. In this case the Helly number becomes $4$, and $x^*(\F)$ is defined in terms of triplets of convex polyhedra with non-empty common intersection.
\begin{theorem}
Let $(p,q)$ be an admissible pair for $d=3$ and let $\F$ be a family of compact convex polyhedra in $\R^3$ with a total of $n$ vertices and the $(p,q)$-property. We can compute a set of at most $p-q+1$ points stabbing $\F$ in $\bigO((p-q+1)n^{5/2}\log^{10} n(\log\log n)^{1/6} + np^3)$ expected time.
\end{theorem}

\begin{proof}
First, we compute the family $\F_2$ of all polyhedra obtained as pair-wise intersections of polyhedra in $\F$. This can be done in $\bigO(n^2)$ time using a linear-time algorithm to compute each intersection (for example the one in \cite{Chan2016}). Assuming the planes defining the polyhedra are in general position and none of the edges lie in a plane parallel to the $yz$-plane, the leftmost point in the intersection of three polyhedra in $\F$ is either
\begin{enumerate}
    \item the leftmost point of a polyhedron in $\F_2$ contained in the interior of a polyhedron of $\F$,
    \item the leftmost point of a polyhedron in $\F$ contained in the interior of a polyhedron of $\F_2$,
    \item the intersection of an edge of a polyhedron in $\F_2$ and the interior of a facet of a polyhedron in $\F$ (note that not all such intersections define the leftmost point of the intersection of three polyhedra in $\F$).
\end{enumerate}
The rightmost point corresponding to the first two cases can be found in expected time $\bigO(n^{9/4}\log^{\bigO(1)}n)$ by triangulating all polyhedra and then using the same methods as for the 2D case.
We now focus on the third case. In what follows, we only deal with triangular facets, as the general case reduces to this one by triangulating all facets in $\bigO(n^2)$ total time. We will preprocess the triangles and then query each edge to count the number of triangles which intersect it and define the leftmost point of a three-wise intersection of polyhedra. Testing if a segment intersects a triangle in $\R^3$ can be done by comparing the signs of three polynomial functions of degree three on the coordinates of the points (see for example \cite{Segura1998}). If $e=(p,q)$ is the segment we are testing against a triangular facet $f$, these polynomials take the form of the following determinant, where $(a,b)$ is one of the three edges of $f$:
\begin{align*}
D(e,f) = 
\begin{vmatrix}
    p_x & p_y & p_z & 1\\
    q_x & q_y & q_z & 1 \\
    a_x & a_y & a_z & 1 \\
    b_x & b_y & b_z & 1
\end{vmatrix}.
\end{align*}
It can be checked that testing $D(e,f) \geq 0$ can be expressed as testing if $P(f) \in H(e)$, where $P(f)$ is a point in $\R^5$ depending only on $f$ and $H(e)$ is a halfspace of $\R^5$ depending only on $e$. The most convenient way is perhaps to use the algorithm described in section 5.2 of \cite{Agarwal_survey} (which is a very slight variant of the algorithm in \cite{Agarwal1994}), which computes a linearization of smallest dimension and simply involves computing the rank of a matrix whose coefficients depend on those of the polynomial to linearize.

In \cite{AGARWAL2002}, it is further shown that when given an edge $e$ of a polyhedron and a facet $f$ of another polyhedron such that $e$ and $f$ intersect, testing if $e\cap f$ is the leftmost point of the intersection of the corresponding polyhedra can be expressed as testing if the outward normal vector of $f$ lies in the intersection of three halfspaces (depending on $e$ and the faces which support it). 

Using again \cref{thm:range-query}, this time in dimension $d=5$ and with $k=6$, we can preprocess the $\bigO(n^2)$ facets of $\F_2$ in $\bigO(n^{5/2}\log^{10} n(\log\log n)^{1/6})$ time such that we can query any oriented edge in $\bigO(n^{3/2}\log^{10} n(\log\log n)^{1/6})$ expected time. By querying all $\bigO(n)$ edges in $\F$ we can detect a leftmost point in the intersection of three polyhedra in $\F$ corresponding to the third case exists in $\bigO(n^{5/2}\log^{10} n(\log\log n)^{1/6})$ total expected time. By applying \cref{lemma:chan} as we did in the planar case, we thus get the result.
\end{proof}

\part*{Other conditions}
In this part we return to families of convex sets in the plane and investigate further conditions that are sufficient for the family to be stabbed by a fixed number of points.
In the whole part, we will consider the planar case, but we expect that with some more care the results can be extended to higher dimensions.

The main reason why we consider the planar case is the following: all proofs below use the algorithm given above, all we do is showing that the algorithm is also correct under some assumptions different than the $(p,q)$-condition.
In particular, we immediately get efficient algorithms for the below results.

\section{Holes}

The first condition we investigate considers \emph{holes} in the union of sets.
Let $\F$ be a finite family of convex sets in the plane and let $A\subset\mathbb{R}^2$ be the union of the sets in $\F$.
A hole is a bounded connected component of $\mathbb{R}^2\setminus A$.

There is an equivalent formulation of Helly's theorem due to Breen, which in the plane can be stated as follows: let $\F$ be a finite family of pairwise intersecting convex sets in the plane with the property that the union of any three of them has no hole, then $\F$ can be stabbed by a single point \cite{Breen, Montejano}.
We prove the following generalization of this result.

\begin{theorem}
\label{thm:holes}
Let $\F$ be a finite family of pairwise intersecting convex sets in the plane with the property that the union of any $k+3$ of them has at most $k$ holes, then $\F$ can be stabbed by $k+1$ points.
Further, the $k+1$ stabbing points can be chosen to lie on a single line.
\end{theorem}

\begin{proof}
Let $x^*(\F)$ be as above, that is, the lexicographical maximum among any lexicographical minimums in the intersection of two sets in $\F$, and let $F_1$ and $F_2$ be the sets in $\F$ that define $x^*(\F)$.
Consider the vertical line $v$ through $x^*(\F)$ and let $F'_1$ and $F'_2$ be the parts of $F_1$ and $F_2$, respectively, that lie to the left of $v$.
Let now $\ell$ be a line through $x^*(\F)$ which separates $F'_1$ and $F'_2$.
Such a line exists as otherwise $x^*(\F)$ would not be the lexicographical minimums in the intersection of $F'_1$ and $F'_2$.
Further note that any set in $\F$ that is not stabbed by $x^*(\F)$ must intersect $\ell$ to the left of $x^*(\F)$: there cannot be intersections exclusively to the right of $x^*(\F)$ by its definition, and as any set intersects $F_1$ and $F_2$, it follows from convexity that it must also intersect $\ell$.

Let now $\mathcal{R}$ be the family of remaining sets, that is, the sets not stabbed by $x^*(\F)$.
We claim that among any $k+1$ of them, some two intersect along $\ell$.
Indeed, if there were $k+1$ sets whose intersections with $\ell$ are pairwise disjoint, the union of these sets with $F_1$ and $F_2$ would have $k+1$ holes, which is excluded by the assumptions of the theorem.
We can thus apply the Hadwiger-Debrunner $(p,q)$-theorem on $\ell$ to stab $\mathcal{R}$ with $k$ points, so in total we have stabbed $\F$ with $k+1$ collinear points.
\end{proof}

Note that opposed to the proof of the Hadwiger-Debrunner $(p,q)$-theorem, we only compute $x^*(\F)$ once.
After this, we only need the 1D-variant, where stabbing points of $n$ intervals can easily be computed in time $O(n\log n)$.
We thus get the following.
\begin{proposition}
Let $\F$ be a family of compact convex polygons in the plane with a total of $n$ vertices and with the property that the union of any $k+3$ of them has at most $k$ holes. We can compute a set of at most $k+1$ collinear points stabbing $\F$ in $\bigO(n^{4/3}\log^{8} n(\log\log n)^{1/3})$ expected time.
\end{proposition}

\section{Number of intersections}

Another result to which our algorithm can be applied is the following, due to Montejano and Sober\'{o}n \cite{MontejanoSoberon}.

\begin{theorem}
Let $p,q,r,d$ be integers with $p>q>d$ and $r>\binom{p}{q}-\binom{p+1-d}{q+1-d}$.
Let $\F$ be a family of convex sets in $\mathbb{R}^d$ with the property that for any $p$ of them at least $r$ of their $q$-tuples intersect.
Then $\F$ can be stabbed by $p-q+1$ points.
\end{theorem}

In the plane, their proof is analogous to our proof of Theorem \ref{thm:holes}, that is, after removing all sets already stabbed by $x^*(\F)$, the assumptions on the set can be used to show that the intersections of the remaining sets with the dividing line $\ell$ satisfy the $(p-2,q-1)$-property.
In particular, the same algorithm to compute the stabbing points is correct, and we get the following.

\begin{proposition}
Let $p,q,r$ be integers with $p>q>2$ and $r>\binom{p}{q}-\binom{p-1}{q-1}$.
Let $\F$ be a family of compact convex polygons in the plane with a total of $n$ vertices and with the property that for any $p$ of them at least $r$ of their $q$-tuples intersect. We can compute a set of at most $p-q+1$ collinear points stabbing $\F$ in $\bigO(n^{4/3}\log^{8} n(\log\log n)^{1/3})$ expected time.
\end{proposition}

\part*{Ordered-Helly systems}

\section{An axiomatic approach}
We have seen in the last chapter an approach that uses Helly's theorem to prove the Hadwiger-Debrunner theorem.  A natural path forward is to try adapting the method to other contexts where Helly-type theorems exist and prove corresponding $(p,q)$ versions. By taking a close look at our proof for the Hadwiger-Debrunner $(p,q)$-theorem, we can observe that we made use of relatively few properties of compact convex sets. These properties are (i) closure under intersection, (ii) existence of a lexicographically minimum point, (iii) Helly's theorem as well as (iv) the fact that the set of all points lexicographically smaller than some point $y$ is convex (this last property doesn't appear explicitly but is needed in the proof of \cref{lemma:lexmin}). We define \emph{Ordered-Helly systems} as set systems with analogous properties.
\begin{definition}[Ordered-Helly system]$ $\\
An Ordered-Helly system $\mathfrak{S}$ is a tuple $(\mathcal{B},\mathcal{C},\mathcal{D},h,\preceq)$ consisting of
\begin{itemize}
\item a set $\mathcal{B}$, called the \emph{base-set};
\item a family $\mathcal{C}$ of subsets of $\mathcal{B}$,  whose members are called \emph{convex sets} or \emph{$\mathfrak{S}$-convex sets};
\item a family $\mathcal{D} \subset \mathcal{C}$ whose members are called \emph{compact sets} or \emph{$\mathfrak{S}$-compact sets};
\item a total order $\preceq$ on $\mathcal{B}$;
\item and an integer $h\geq 2$, called the \emph{Helly-number} of $\mathfrak{S}$
\end{itemize} 
with the following properties.
\begin{enumerate}
\item (Intersection closure)\\$\mathcal{D}$ is closed under intersection, i.e. for all $S_1, S_2 \in \mathcal{D}$ we have $S_1\cap S_2 \in \mathcal{D}$.
\item (Attainable minimum)\\For all non-empty $S \in \mathcal{D}$, there exists $x\in S$ such that for all $y \in S$, $x\preceq y$. This $x$ is necessarily unique and we call $x$ the $\preceq$-min of $S$. We define the $\preceq$-max of a set similarly, if it exists.
\item (Convex order)\\For all $t \in \mathcal{B}$, we have $\{x\in \mathcal{B}\mid x \preceq t\ \text{and}\ x \neq t\} \in \mathcal{C}$.
\item (Helly property)\\If $\F \subset \mathcal{C}$ is a finite subset of $n\geq h$ sets of $\mathcal{C}$ such that every subfamily of $h$ members of $\F$ has a non-empty common intersection, then all members of $\F$ have a non-empty common intersection.
\end{enumerate}

\end{definition}

As was stated earlier, this structure is enough to carry out a similar proof as the one we saw for the Hadwiger-Debrunner theorem. Call a pair $(p,q)$ of integers $h$-admissible if $p\geq q \geq h$ and $(h-2)p<(h-1)(q-1)$. Then we have the following.
\begin{theorem} \label{thm:sysHD}
Let $\mathfrak{S} =(\mathcal{B},\mathcal{C},\mathcal{D},h,\preceq)$ be an Ordered-Helly system. Let $(p,q)$ be an $h$-admissible pair of integers. Let $\F$ be a finite family of non-empty sets of $\mathcal{D}$. Suppose that $\F$ has the $(p,q)$-property. Then there exist $p-q+1$ elements of $\mathcal{B}$ stabbing $\F$.
\end{theorem}

It should be mentioned that the existence of a Helly number alone is not enough to show such a result, see \cite{Alon2002} for an example of a set system with Helly number 2 but no general $(p,q)$-theorem.

To prove this theorem, we will make use of an analogous to \cref{lemma:lexmin}. Let us state and prove this analogous lemma.

\begin{lemma}
\label{lemma:sysmin}
Let $\mathfrak{S} = (\mathcal{B},\mathcal{C},\mathcal{D},h,\preceq)$ be an Ordered-Helly system.
Let $\F\subset \mathcal{D}$ be a family of $n\geq h$ sets in $\mathcal{D}$ such that $I := \bigcap\F$ is non-empty. Let $x$ be the $\preceq$-min of $I$ (which exists by the properties of intersection closure and attainable minimum). Then, there exists a subfamily $\mathcal{G} \subset \F$ of size $h-1$ such that $x$ is the $\preceq$-min of $\mathcal{G}$.
\end{lemma}
\begin{proof} Let $\F$, $I$ and $x$ be as specified in the statement.\\ 
Let $S_x$ denote ${\{y\in \mathcal{B}\mid y \preceq x\ \text{and}\ y \neq x\}}$, which is a $\mathfrak{S}$-convex set by the property of convex order. It is also disjoint from $I$ as $\preceq$ is a total order. By the Helly property, there exists a subfamily of $h$  members of $\F\cup\{S_x\}$ with an empty common intersection. These members have to include $S_x$, as all members of $\F$ have a non-empty common intersection. Let $\mathcal{G} \subset \F$ be the family consisting of the remaining $h-1$ sets and let $x_\mathcal{G}$ be the $\preceq$-min of $I':=\bigcap\mathcal{G}$ (which is a non-empty $\mathfrak{S}$-compact set). We know that $x_\mathcal{G} \preceq x$ because $x \in I'$. If we now suppose $x_\mathcal{G} \neq x$ this implies that $x_\mathcal{G} \in S_x$ and contradicts the fact that $I'\cap S_x = \emptyset$. Thus, $x_\mathcal{G} = x$.
\end{proof}

We are now ready to prove \cref{thm:sysHD}.

\begin{proof}[of \cref{thm:sysHD}]
Let $\mathfrak{S} = (\mathcal{B},\mathcal{C},\mathcal{D},h,\preceq)$ be an Ordered-Helly system.

Let $(p,q)$ be an $h$-admissible pair, and let $\F$ be a family of sets of $\mathcal{D} $ with the $(p,q)$-property. 

We will reason by induction on $p$, the base case being $p=q=h$ which is true by the Helly property of the system. If $p=q>h$, then $\F$ also has the $(h,h)$ property (as having the $(p,q)$-property implies having the $(p-1,q-1)$-property) and the results again follows from the Helly property of the system.

So suppose that $p>q$ and that the result is true for any $h$-admissible pair $(p',q')$ with $p'<p$. 

If $(h-2)p=(h-1)(q-1)-k-1$ for $k\geq 1$, then notice that $(p-k,q-k)$ is an $h$-admissible pair, as in that case $(h-2)(p-k) = (h-1)(q-k-1)-1$ which together with $p>q$ also implies that $q-k \geq h$. Thus the result follows from the induction hypothesis. 

It now remains to consider the case where $p>q$ and $(h-2)p = (h-1)(q-1)-1$. 

To do so, let us construct an element $b^*(\F)$ as follows: 
\begin{itemize}
\item For every non-empty subfamily $\SIt \subset \F$ of $h-1$ $\mathfrak{S}$-convex sets with non-empty intersection, let $b_{\SIt}$ be the $\preceq$-min of $I_{\SIt} = \bigcap_{C\in\SIt}C$, which exists by the properties of intersection closure and attainable minimum.
\item Let $b^*(\F)$ be the $\preceq$-max element among all such $b_{\SIt}$'s.
\end{itemize}

Let $\mathcal{G}$ be one of the families defining $b^*(\F)$, that is, $\mathcal{G} \subset \F'$ is a subfamily of $h-1$ sets which have $b^*(\F)$ as the $\preceq$-min of their intersection. 

To establish the theorem, it is enough to show that by choosing $b^*(\F)$ as one of our stabbing elements, we can stab all the remaining sets (i.e. those which do not contain $b^*(\F)$) with $p-q$ elements. Let $\mathcal{R} = \{S \in \F\ |\ b^*(\F) \not\in S\}$ be the set of remaining sets.

Let us argue that for any $S\in \mathcal{R}$, $S \cap (\bigcap\mathcal{G})$ is empty. To do so, suppose it was not, and let $y$ be the $\preceq$-min of that intersection. By \cref{lemma:sysmin}, $y$ is the $\preceq$-min of the intersection of $h-1$ sets in $\F$. Moreover, by definition of $\mathcal{R}$ and $\mathcal{G}$, $b^*(\F) \prec y$. This contradicts the definition of $b^*(\F)$. Thus, $S \cap (\bigcap\mathcal{G})$ is empty.

Two cases arise:
\begin{enumerate}
\item $(|\mathcal{R}| \geq p-h+1)$ We show that $\mathcal{R}$ has the $(p-h+1,q-h+2)$-property. Indeed, choose any $p-h+1$ members from $\mathcal{R}$ together with the $h-1$ members from $\mathcal{G}$. We know from the $(p,q)$-property of $\F$ that there exists a subfamily $\mathcal{E} \subset \mathcal{R}\cup \mathcal{G}$ of size $q$ whose members have a non-empty common intersection. $\mathcal{E}$ cannot contain all elements of $\mathcal{G}$, as $q>h-1 = |\mathcal{G}|$ and the intersection of all members of $\mathcal{G}$ together with any member of $\mathcal{R}$ is empty. Thus, $\mathcal{E}$ contains at least $q-h+2$ members of $\mathcal{R}$. This shows that $\mathcal{R}$ has the $(p-h+1,q-h+2)$-property. Notice that with the assumptions $p>q$ and $(h-2)p = (h-1)(q-1)-1$ which we are working under, $(p-h+1, q-h+2)$ is admissible. Thus, by the induction hypothesis, $\mathcal{R}$ can be stabbed with $p-h+1 - (q-h+2) + 1 = p-q$ elements of $\mathcal{B}$.
\item $(|\mathcal{R}| < p-h+1)$ In this case, choose $\mathcal{R}$ as a whole together with $\mathcal{G}$ and $p-h+1-|\mathcal{R}|$ other members of $\F$. By the same reasoning as in case 1., there exists a subset of $q-(h-2+p-h+1-|\mathcal{R}|) = |\mathcal{R}| + 1 + q - p$ members of $\mathcal{R}$ which intersect and can thus be stabbed by a single element of $\mathcal{B}$. The remaining $|\mathcal{R}| - (|\mathcal{R}| + 1 + q - p) = p-q-1$ sets can trivially be stabbed by $p-q-1$ elements.
\end{enumerate} 
Thus, $\mathcal{R}$ can be stabbed by $p-q$ elements, which implies that $\F$ can be stabbed by $p-q+1$ elements and by induction, concludes the proof.
\end{proof}

\begin{remark-nn}
We could relax the properties of an Ordered-Helly system somewhat and still be able to prove this theorem. In particular, properties $1.$, $2.$ and $3.$ could be replaced by the two following properties, which already hold in an Ordered-Helly system.
\begin{itemize}
\item For all $1\leq k \leq h$ and all $S_1, S_2, \ldots S_{k} \in \mathcal{D}$ with a non-empty common intersection there exists some $x\in \bigcap_i S_i$ such that for all $y\in \bigcap_i S_i$, $x\preceq y$.
\item If $\F \subset \mathcal{D}$ is a family of $h$ compact sets and $x$ is the $\preceq$-min of $\bigcap \F$, then there exists some subfamily $G$ of size $h-1$ such that $x$ is also the $\preceq$-min of $\bigcap\mathcal{G}$.
\end{itemize}

In fact, none of the conditions $1.$ to $3.$ are necessary in the sense that there exist families of sets violating all three for which a Hadwiger-Debrunner type theorem does hold. Consider for example the family of all open disks in the plane with the lexicographical order on points. Neither intersection closure, attainable minimum nor convex order holds in this case, but of course the Hadwiger-Debrunner theorem still applies as these are a special case of convex sets.

However, condition $3.$ (convex order) is in fact necessary in the sense that dropping it while maintaining the others unchanged would make \cref{thm:sysHD} false. Otherwise, we could for example prove that a family of axis aligned rectangles in the plane with the $(3,2)$ property can be stabbed with two points. This is false, as the following example illustrates.

\begin{figure}[h]
\centering
\scalebox{1}{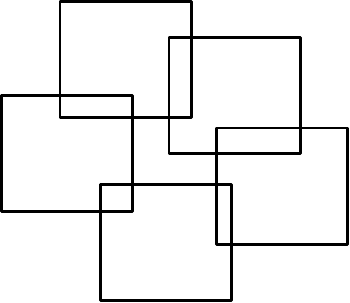}
\caption{Counterexample to the claim that axis aligned rectangles in the plane with the $(3,2)$-property can always be stabbed with two points.}
\label{fig:rectangles}
\end{figure}

\end{remark-nn}

\section{Computing stabbing points in an Ordered-Helly system}

The proof we saw once again leads to an algorithm computing stabbing elements of a family of $\mathfrak{S}$-compact sets with the $(p,q)$-property for an admissible pair $(p,q)$, given we have access to some oracles. We will write the run-times in terms of the description complexity of a set, which depends on the exact context. Thus, for a $\mathfrak{S}$-compact set $S$, let $\#S$ denote this complexity (of at least $1$), and for a family $\F$ of $\mathfrak{S}$-compact sets, let $\#\F := \sum_{S\in\F}\#S$. 

Consider an Ordered-Helly system $\mathfrak{S} = (\mathcal{B},\mathcal{C},\mathcal{D},h,\preceq)$ (for a constant $h$) and suppose we have access to the following oracles.
\begin{itemize}
\item For two elements $b_1, b_2 \in \mathcal{B}$, we can test $b_1 \preceq b_2$ in constant time.
\item For a family of at most $h-1$ $\mathfrak{S}$-compact sets $\F \subset \mathcal{D}$, we can test if the sets in $\F$ have a common intersection and compute the $\preceq$-min of that intersection if it is non-empty in $\bigO(\#\F)$ time.
\item For a $\mathfrak{S}$-compact set $S \in \mathcal{D}$ and a point $b \in \mathcal{B}$ we can test if $b\in S$ in $\bigO(\#S)$ time.
\end{itemize}

We could naturally consider other run-times for these oracles. We only specify them in order to showcase an example of run-time analysis which is tighter than if we had worked with general run-times and swapped in concrete functions afterwards (and matches the case of convex polytopes in $\R^d$ for small $d$). Other run-times might require other specialized forms of analysis.

Now, let $\F \subset \mathcal{D}$ be a family of $\mathfrak{S}$-compact sets. Among all points in $\mathcal{B}$ defined as the $\preceq$-min of the intersection of $h-1$ sets in $\F$, let $b^*(\F)$ be the $\preceq$-max of those. 

Let us state two lemmas which will be useful for our algorithm. The proofs are identical to those of \cref{lemma:d-wise-intersect} and \cref{lemma:base-stab-poly} and we repeat them only for the readers convenience.
\begin{lemma}\label{lemma:h-wise-intersect}
We can compute $b^*(\F)$ in $\bigO(\#\F\cdot|\F|^{h-2}))$ time.
\end{lemma}
\begin{proof}
We can compute $b^*(\F)$ by testing for intersection in every subfamily $\mathcal{G}$ of $\F$ of size $h-1$ and computing the $\preceq$-min of that intersection if it is non-empty. 

If we consider some fixed subfamily $\mathcal{G}$, the computation for that subfamily will cost at most $c\cdot\#\mathcal{G}$ for some constant $c$ which doesn't depend on $\mathcal{G}$. Charge this cost to the sets $S \in \mathcal{G}$ by attributing a cost of  $c\cdot\#S$ to a set $S$. 

Now, consider the cost charged to some fixed set $S$ for the whole computation. As $S$ appears in no more than $|\F|^{h-2}$ subfamilies of size $h-1$, its total cost charge is upper bounded by $c\cdot\#S\cdot |\F|^{h-2}$. Summing across all sets $S\in\F$, we get a total cost of $\bigO(\#\F\cdot|\F|^{h-2}))$.
\end{proof}

\begin{lemma}\label{lemma:baseStab}
Suppose there exists some subfamily $\mathcal{G} \subset \F$ of size $k+1$ such that all sets in $\mathcal{G}$ have a common intersection, where $k$ is a known parameter. We can compute $|\F|-k$ points in $\mathcal{B}$ stabbing $\F$ in  $\bigO(\#\F\cdot |\F|^{h-1})$ time.
\end{lemma}
\begin{proof}
If $k+1\leq h-1$, then we can test every subfamily of size $k+1$ for common intersection and compute its $\preceq$-min for a total cost of $\bigO( \#\F\cdot |\F|^{k+1}) \leq \bigO( \#\F\cdot |\F|^{h-1})$.

If $k+1> h-1$, then we know from \cref{lemma:sysmin} that the $\preceq$-min of the intersection of all sets in $\mathcal{G}$ is also the $\preceq$-min of the intersection of some $h-1$ sets in $\F$. Thus, one can find a point stabbing at least $k+1$ sets by computing the $\preceq$-min point for each subfamily of size $h-1$ (in $\bigO(\#\F|\F|^{h-1})$ time) and counting the number of sets intersected for each of the $\bigO(|\F|^{h-1})$ computed points (in $\bigO(\#\F|\F|^{h-1})$ time as well).

As soon as we find a point $b$ stabbing at least $k+1$ sets, we return $b$ along with the $\preceq$-min of every set in $\F$ which is not stabbed by $b$.
\end{proof}

With these algorithms, we can now prove the following.

\begin{theorem}\label{thm:stabOH} 
Let $\F$ be a family of $\mathfrak{S}$-compact sets with the $(p,q)$-property. Suppose we have access to the relevant oracles described above. We can compute a set of at most $p-q+1$ elements stabbing $\F$ in time
\[\bigO((p-q+1)(\#\F)^{h-1} + (\#\F)p^{h-1}).\]
\end{theorem}
\begin{proof}$ $\\
Consider the following algorithm.
\begin{enumerate}
\item Reduce $p$ and $q$ (as done in the proof of \cref{thm:sysHD}) to reach the case where $p=q=h$ or the case where $p>q$ and $(h-2)p = (h-1)(q-1)-1$. 
\item Construct an element $b^*(\F)$ defined as in the proof. We choose it as one of our stabbing points.
Now, remove from $\F$ all the sets which are stabbed by this point. If there are any remaining sets then either $|\F| \geq p-h+1$ and $\F$ satisfies the $(p-h+1,q-h+2)$-property, where $(p-h+1,q-h+2)$ is $h$-admissible, or $\F$ consists of $p-q+k$ sets, $k < q-h+1$, where some $k+1$ of them have a common intersection.
\item In the first case, we can continue inductively.
\item In the second case we can trivially stab the remaining sets using $p-q$ elements.
\end{enumerate}

$\bullet$ \textit{Correctness:}\\
The correctness of the algorithm follows from the proof of \cref{thm:sysHD}. The only detail that needs some additional scrutiny is the correctness for the base case $p=q=h$. Notice that in this case all sets have a common intersection and \cref{lemma:sysmin} ensures that $b^*(\F)$ stabs the whole family $\F$.

$\bullet$ \textit{Runtime:}\\
We know that Step 1 can be done in $\bigO((\#\F)^{h-1})$ time and has to be done at most $(p-q+1)$ times. We also know from \cref{lemma:baseStab} that Step 4 can be done in $\bigO((\#\F)^{h-1} + (\#\F)p^{h-1})$ time. This step is only done once.

Thus, we get a total runtime of 
\[\bigO((p-q+1)(\#\F)^{h-1} + (\#\F)p^{h-1}).\]
\end{proof}

With access to the right oracle, we could for example apply \cref{lemma:chan} analogously to what we did for convex polytopes in the Euclidean setting and get the corresponding speedup.

\part*{Examples of Ordered-Helly systems}

Until now, the only Ordered-Helly system we have seen is the one corresponding to compact convex sets in $\R^d$. We will see that this structure does have some other interesting representatives and is not restricted to this single example (in which case the usefulness of introducing it would have been doubtful).

\section{Hadwiger-Debrunner type results for subsets of \texorpdfstring{$\R^d$}{Rd}}

Let us start by stating and proving some Hadwiger-Debrunner type results for sets which are defined as the intersection of a compact convex set in $\R^d$ with a subset $S\in \R$. To do so define the $S$-Helly number as follows.

\begin{definition}
Let $S$ be a subset of $\R^d$. We call $S$-Helly number and write $h(S)$ the smallest integer $k>0$ such that the following holds:\\
Given a finite family $\F$ of convex sets in $\R^d$, if in every subfamily of $\F$ of size $k$ all sets share a point in $S$, then all sets in $\F$ share a point in $S$.\\
If no such $k$ exists, then $h(S) = \infty$.
\end{definition}

One of the first results concerning $S$-Helly numbers was discovered by Doignon \cite{Doignon1973}, and is the case $S=\Z^d$.
\begin{theorem}[Doignon]
Let $\F$ be a family of $n\geq 2^d$ convex sets in $\R^d$. If in every subfamily of $\F$ of size $2^d$ all sets share a point in $\Z^d$, then all sets in $\F$ share a point in $\Z^d$.
\end{theorem}

In \cite{Hoffman2006} Hoffman stated a mixed-integer version of this theorem, which generalizes both Helly's theorem and Doignon's version. It was later rediscovered and proved in detail by Averkov and Weismantel \cite{Averkov2012}.
\begin{theorem}[Mixed-Integer Helly] \label{thm:mixed-helly}
Let $\F$ be a family of $n\geq (d+1)2^k$ convex sets in $\R^{d+k}$, where $d,k\geq 0$ and $d+k \geq 1$. If in every subfamily of $\F$ of size $(d+1)2^k$ all sets share a point in $\R^d\times\Z^k$, then all sets in $\F$ share a point in $\R^d\times\Z^k$.
\end{theorem}

Let us show that a corresponding Hadwiger-Debrunner-type theorem holds.

\begin{theorem}[Mixed-Integer Hadwiger-Debrunner]
$ $\\
Let $d,k \geq 0$ be integers such that $d+k\geq1$. Let $(p,q)$ be a $(d+1)2^k$-admissible pair. Let $\F$ be a finite family of sets obtained as the intersection of $\R^d\times\Z^k$ with a compact convex set in $\R^{d+k}$. Suppose that $\F$ has the $(p,q)$ property. Then there exist $p-q+1$ points in $\R^d\times\Z^k$ stabbing $\F$.
\end{theorem}
\begin{proof}
Let $\mathcal{B} = \R^d\times\Z^k$, let $\mathcal{C}$ be the family of all sets obtained as the intersection of $\mathcal{B}$ with a convex set in $R^{d+k}$ and let $\mathcal{D}$ be the family of all sets obtained as the intersection of $\mathcal{B}$ with a compact convex set in $R^{d+k}$. Then, using \cref{thm:mixed-helly}, it is easy to check that $(\mathcal{B},\mathcal{C},\mathcal{D},(d+1)2^k, \leq_{lex})$ is an Ordered-Helly system. Thus, using \cref{thm:sysHD}, we get the result.
\end{proof}

More generally, every upper bound on an $S$-Helly number leads to a corresponding Hadwiger-Debrunner version if $S$ is closed in $\R^d$. The corresponding algorithmic results also follow, provided we have access to the required oracles. 

\section{Abstract convex geometries}

Let us now explore how the structure of Ordered-Helly systems relates to the structure of abstract convex geometries as introduced by  Edelman and Jamison \cite{Edelman1985}. Convex geometries are an abstraction capturing the basic combinatorial structure of classical convexity in a similar manner to matroids capturing the basic combinatorial properties of linear independence. Convex geometries appear in many contexts outside of convex sets such as graph theory or order theory. We refer the interested reader to \cite{Edelman1985} or to Chapter III of \cite{Korte1991} for an in-depth overview. We will only go over the basic definitions and theorems needed for our purpose, which can all be found in the two sources we just mentioned.

\subsection{Some background}
For the following definitions, it is useful to imagine the operator $\tau$ as analogous to the convex hull operator on a point set. 
\begin{definition}
Consider some finite set $E$ and a family $\mathcal{N}$ of subsets of $E$. Let $\tau$ be the operator defined on subsets of $E$ as $\tau(A) = \bigcap\{X \mid A\subset X,\ X\in \mathcal{N}\}$. We say that $(E,\mathcal{N})$ is a convex geometry if it has the following properties.
\begin{enumerate}
\item $\emptyset \in \mathcal{N}$, $E \in \mathcal{N}$.
\item $X,Y \in \mathcal{N}$ implies $X\cap Y \in \mathcal{N}$.
\item If $y,z\not\in\tau(X)$ and $z\in\tau(X\cup \{y\})$ then $y\not\in \tau(X\cup \{z\})$. 
\end{enumerate}
The sets in $\mathcal{N}$ are called convex.
\end{definition}

Extreme points are defined in the same way as in the Euclidean setting:
\begin{definition}
For a set $A \subset E$, we say that $x\in A$ is an extreme point of $A$ if $x\not\in \tau(A \setminus \{x\})$. The set of extreme points of $A$ is denoted by $ex(A)$.

A set $X\subset E$ is called free if $X=ex(X)$.
\end{definition}

We will use the following concept.
\begin{definition}
A sequence $x_1,\ldots,x_k$ of points of $E$ is called a shelling sequence if for all $1\leq i \leq k$, $x_i$ is an extreme point of $E\setminus \{x_1,\ldots,x_{i-1}\}$.
\end{definition}

A shelling sequence can be thought of as a way to reach a convex set by starting with the whole set $E$ and stripping away points one after the other in such a way that the set remains convex at each step. A useful characterisation of convex sets for our purpose is the following, where we describe a convex set via a shelling process.
\begin{proposition}[\cite{Edelman1985}]
A set $X \subset E$ is convex if and only if there exists a shelling sequence $x_1,\ldots,x_k$ such that $X = E\setminus \{x_1,\ldots,x_k\}$.
\end{proposition}

The final ingredient we need is the following Helly-type theorem for convex geometries.

\begin{theorem}[\cite{Edelman1985}]
\label{thm:abstracthelly}
Let $h(\mathcal{N})$ denote the smallest integer $k$ such that the following holds:

For a family $\F$ of convex sets, if every subfamily of size at most $k$ has a non-empty intersection, then $\F$ has a non-empty intersection.

Then $h(\mathcal{N})$ is equal to the maximum size of a free convex set.
\end{theorem}

\subsection{Hadwiger-Debrunner theorem for convex geometries}

We are now ready to state and prove a Hadwiger-Debrunner-type theorem for convex geometries.

\begin{theorem} \label{thm:abstractHD}
Consider a convex geometry $(E,\mathcal{N})$. Let $h$ be the size of a maximal free convex set and let $(p,q)$ be an $h$-admissible pair. Let $\F\subset \mathcal{N}$ be a family of $n\geq p$ non-empty convex sets. If $\F$ has the $(p,q)$-property then there exist $p-q+1$ elements of $E$ stabbing $\F$.
\end{theorem}
\begin{proof}
We know that $\emptyset$ is convex, thus there exist a shelling sequence $S=\{x_1, \ldots, x_k\}$ such that $\emptyset = E\setminus S$, i.e. $S=E$. Let for $1\leq i,j \leq k$, let us say that $x_i \preceq x_j$ if and only if $i\geq j$. Let $1\leq t \leq k$ be an integer. Because $\{x_1,x_2,\ldots, x_{t-1}\}$ is a valid shelling sequence, $\{x\in E \mid x\leq x_t\}$ is a convex set. Thus, $\preceq$ has the convex order property.

Let $h$ be the maximum size of a free convex set. Then, it is easy to verify that $\mathfrak{S}=(E,\mathcal{N},\mathcal{N}, h, \preceq)$ also has the intersection closure and attainable minimum properties. The Helly property (for Helly number $h)$ is given by \cref{thm:abstracthelly}.

Thus, $\mathfrak{S}$ is an Ordered-Helly system and we get the result from \cref{thm:sysHD}.
\end{proof}

\subsection{Two examples of convex geometries}

We will now give two illustrative examples of abstract convex geometries and the resulting Hadwiger-Debrunner type results we obtain for them. One such convex geometry (arguably the most natural) is the one obtained by taking convex hulls of subsets on a finite point set in $\R^d$. This is conceptually similar to the case of polytopes in Euclidean space which we have already discussed. The following two examples are perhaps not so immediately related.

\paragraph*{Subtrees of a tree}
\begin{proposition}[\cite{Edelman1985}]
Let $T$ be a tree on a set of vertices $V$. Let $\mathcal{N}$ be the family of all sets of vertices corresponding to subtrees of $T$. Then $(V,\mathcal{N})$ is a convex geometry with Helly number $h(\mathcal{N}) = 2$.
\end{proposition}

This means that for a given tree $T$ and a given family of subtrees of $T$, if all pairs of subtrees intersect at some vertex, then all subtrees share a vertex. Using \cref{thm:abstractHD} we can thus get the following result.

\begin{corollary}
Let $T$ be a tree and let $\F$ be a family of subtrees of $T$ (represented as sets of vertices). Let $(p,q)$ be a $2$-admissible pair. Let $\F\subset \mathcal{N}$ be a family of non-empty subtrees of $T$ with the $(p,q)$-property. Then $\F$ can be stabbed with $p-q+1$ vertices.
\end{corollary}

From an algorithmic point of view, let us suppose that the tree is represented as a conventional pointer structure and that the subtrees in $\F$ are themselves represented in full as trees. One can compute a shelling sequence of the empty tree (and thus $\preceq$) by starting with the whole tree $T$ and choosing leaves to cut off until we reach the empty tree. This amounts to $O(|V|)$ preprocessing time. 
We can trivially find the $\preceq$-min of a subtree or test if a subtree $S$ contains a vertex in $O(|V(S)|)$ time. 

Using \cref{thm:stabOH}, this leads to an algorithm finding stabbing vertices in $\bigO(|V| + p(\#\F))$ where $\#\F$ is the sum of the number of vertices over all subtrees in $\F$.

\paragraph*{Ideals of a partially ordered set}$ $

For a poset $(E,\leq)$, we say that a set $S\subset E$ is an ideal of $E$ if for all $x\in S$ and all $y\in E$, $y\leq w \Rightarrow y\in S$. Let $width(E)$ denote the maximum size of an antichain in $E$. Then the following holds.
\begin{proposition}[\cite{Edelman1985}]
Let $(E,\leq)$ be a finite poset. Let $\F = \{S\subset E \mid S\ \text{is an ideal}\}$. Then $(E,\F)$ is a convex geometry with Helly number $width(E)$.
\end{proposition}

Using \cref{thm:abstractHD} we can thus get the following result.

\begin{corollary}
Let $(E,\leq)$ be a finite poset. Let $(p,q)$ be an $width(E)$-admissible pair and let $\F$ be a family of non-empty ideals of $E$ with the $(p,q)$-property. Then $\F$ can be stabbed by $p-q+1$ elements of $E$.
\end{corollary}

From an algorithmic point of view, the situation is similar to the one for subtrees of a tree if we choose to represent ideals as the sets of their elements.

\part*{Conclusion}

We have shown how to stab convex polygons with a total of $n$ vertices and the $(p,q)$-property (for admissible $(p,q)$) in expected $\tilde{\bigO}(n^{4/3})$ time with respect to $n$. As an intermediate step, we compute a certain quantity $x_2^*$, which is a Hopcroft-Hard problem, in $\tilde{\bigO}(n^{4/3})$ expected time. While this is believed to be near optimal, finding a non-trivial lower-bound for the original problem remains open. For the 3D case, we have an algorithm running in expected $\tilde{\bigO}(n^{5/2})$ time with respect to $n$.

We have also considered other conditions which allow to conclude that a set of polygons can be stabbed with a fixed number of points, and applied our algorithm to those. One of these conditions is a new generalization of Helly's theorem in the plane in terms of holes in the union of convex sets.

Finally, we have derived $(p,q)$-theorems along with algorithms in other settings where Helly-type theorems are known. An interesting question would be to try deriving other related results in these settings, such as colourful or fractional versions of $(p,q)$-theorems.

A natural next step in the Euclidean setting would be to drop the restriction $(d-1)p < d(q-1)$ and find efficient algorithms for the Alon-Kleitman $(p,q)$-theorem. Their proof of existence of stabbing sets of constant size uses the fractional Helly theorem, whose proof is similar to the above proof of the Hadwiger-Debrunner $(p,q)$-theorem. It is thus conceivable that similar ideas could be applied to this more general case.

\appendix

\newpage
\part*{Appendices}

\section{From compact convex sets to general convex sets in the Hadwiger-Debrunner theorem}
\label{app:compact}
In the main body of this paper, we have proven the Hadwiger-Debrunner theorem for compact convex sets only, claiming that this was done without loss of generality. To see this let us describe how to reduce the general case to this one.

Let $\F = \{S_1, S_2,\ldots, S_n\}$ be a finite family of convex sets in $\R^d$ with the $(p,q)$-property for some admissible pair $(p,q)$. Let us construct a new family of compact convex sets $\F' = \{S'_1, S'_2,\ldots, S'_n\}$ as follows:

\begin{itemize}
\item For every subfamily $\mathcal{G} \subset \F$ of sets with non-empty common intersection, choose $p_\mathcal{G}$ to be a point in that intersection, and let $\mathcal{P}$ denote the set of all such chosen points.
\item For $i\in \{1,\ldots,n\}$, let $S'_i$ be the convex hull of $S_i \cap \mathcal{P}$.
\end{itemize}

It is clear that $\F'$ consists of compact convex sets and that whenever some subfamily $\mathcal{G} \subset \F$ of sets have a common intersection, the corresponding sets in $\F'$ also do. This means that $\F'$ has the $(p,q)$-property. Moreover, for $i\in \{1,\ldots,n\}$ we know that $S'_i \subset S_i$. If we can stab $\F'$ with $p-q+1$ points, the same holds for $\F$. Thus, if the Hadwiger-Debrunner theorem holds for compact convex sets, it also holds for general convex sets.

\section{Adapting the Bentley-Ottmann sweep-line algorithm}
\label{app:bentley-ottmann}
Here we briefly describe how to adapt the Bentley-Ottmann sweep-line algorithm \cite{Bentley1979} to solve the following problem in $\bigO(n^2\log(n))$ time.
\begin{problem}
Given a family of convex polygons in the plane with a total of $n$ vertices, compute a point $p$ stabbing as many polygons as possible.
\end{problem}

Imagine a vertical line sweeping through the plane, stopping each time it reaches the beginning of an edge, the end of an edge or the intersection of two edges, which we call an event (for simplicity, assume the events are separated along the horizontal axis by shifting them infinitesimally ). During the whole sweep, we keep track of the edges crossing our sweep line ordered according to the $y$ coordinate of their intersection with the sweep line. This forms the general idea behind the Bentley-Ottmann algorithm.

More specifically, the algorithm maintains a self-balancing Binary Search Tree (BST) of the segments intersecting the sweep line (this line is conceptual only and is not explicitly stored in any manner) as well as a priority queue of events to come. At each stage neighbouring edges on the sweep line are tested for future intersection and this intersection is added to the priority queue if it exists. Then the next event in the priority queue is considered and the imaginary sweep line is moved to that event. If this is the beginning or the end of an edge, this edge is respectively added to or deleted from the BST. If this event is the intersection of two edges, their positions in the BST are swapped. All these operations can be done in $\bigO(\log(n))$ time with the usual data structures for priority queues and self-balancing BSTs. Thus, as there are at most $\bigO(n^2)$ events, the whole sweep takes $\bigO(n^2\log(n))$ time.

Now, we can partition the edges of the polygons into two classes: those corresponding to the upper hull and those corresponding to the lower hull. We can then augment each node in the BST with the following information: for each node $v$ store the number of leaves which are lower hull edges and upper hull edges in the subtree rooted at $v$. These quantities can easily be maintained in $\bigO(\log(n))$ time per operation performed on the BST.

When considering a new event, we can now easily compute the number of polygons stabbed by the corresponding point by taking the number of upper hull edges above it and subtracting the number of lower hull edges above. Both of these latter quantities can be computed in $\bigO(\log(n))$ time by travelling up the BST from the vertex of interest to the root.

To compute a point $p$ stabbing as many polygons as possible, it is enough to consider only the points which we have defined as events. Thus, with this modified algorithm we can find the event point stabbing the largest number of polygons and thus solve the considered problem in $\bigO(n^2\log(n))$ time.

\section{Proof of \cref{thm:range-query}}
\label{app:range-query}

Here we give pointers to prove \cref{thm:range-query}. The methods used here are standard in the field of geometric range queries and are small variations on proofs found in \cite{Matousek1992} and \cite{Chan2016}. Although we recommend the lecture of these papers for a good understanding of these approaches, we give some details here for the sake of completeness.

Recall the statement of the theorem. 
\thmRangeQuery*

From now on, $d$, $k$ and the mappings $\phi_1,\ldots,\phi_k$ are fixed. For the sake of simplicity we suppose that all objects $s\in S$ we consider are of constant size and that $\phi_i(s)$ can be computed in constant time (but we could precompute all these points as a preprocessing step and then work only with the points $\phi_i(s)$). For any finite set of objects $S$ and any $1\leq k' \leq k'' \leq k$, we let $\phi^{k', k''}_S$ denote the function which maps $(k''-k'+1)$-tuples of halfspaces $H_{k'},\ldots,H_{k''}$ of $\R^d$ to 
\[\phi^{k', k''}_S(H_{k'},\ldots,H_{k''}) := \{s\in S \mid \phi_{k'}(s)\in H_{k'},\ldots,\phi_{k''}(s)\in H_{k''}\}.\]
When $k'=k''$, we also use the notation $\phi^{k'}_S := \phi^{k', k''}_S$.

We need the following result from \cite{Matousek1992}.
\begin{theorem}[\cite{Matousek1992}]\label{thm:matousek_decomp}
For any $1\leq k' \leq k$, any set of $n$ objects $S$ and any parameter $r<n$, we can build a datastructure in $\bigO(nr^{d-1})$ time with the following properties.
\begin{itemize}
    \item There are $t\in \bigO(\log r)$ collections of subsets of $S$, $\mathcal{C}_1, \mathcal{C}_1, \ldots, \mathcal{C}_{t}$ such that $\mathcal{C}_i$ contains $\bigO(\rho^i)$ subsets of size at most $n/\rho^i$ (where $\rho > 1$ is a constant dependent on $r$). We call all the subsets in these collections the inner sets.
    \item There are an additional $\bigO(r^d)$ subsets of $S$, each of size at most $n/r$, called the remainder sets. We denote the collection of these subsets as $\mathcal{R}$.
    \item For any halfspace $H$, we can in $\bigO(\log r)$ time return pointers to $t+1$ of these subsets, $S_1 \in \mathcal{C}_1, S_2 \in \mathcal{C}_2, \ldots, S_{t} \in \mathcal{C}_{t}$ and $S_\mathcal{R}\in \mathcal{R}$, all disjoint, such that 
    \[\phi^{k'}_S(H) = S_1 \cup S_2 \cup \cdots \cup S_{t} \cup \phi^{k'}_{S_\mathcal{R}}(H).\]
\end{itemize}
\end{theorem}
Note that the inner and remainder sets are not stored individually in full but implicitly as a partition tree structure.

We also need this other result from \cite{Chan2016}.
\begin{theorem}[\cite{Chan2016}]\label{thm:chan_composed_query}
For any $1\leq k' \leq k$, any set of $n$ objects $S$, we can build a datastructure in $\bigO(n\log^{k'}n)$ time which can then compute $|\phi^{1,k'}_S(H_{1},\ldots,H_{k'})|$ in $\bigO(n^{(d-1)/d}\log^{k'-1}n)$ expected time for any $k'$-tuple of halfspaces.
\end{theorem}

We can now prove the following, by adapting the proof of Theorem 6.1 in \cite{Matousek1992}.
\begin{theorem}\label{thm:big_space_query}
For any, $1\leq p \leq k$, any set of $n$ objects $S$, we can build a datastructure in time $\bigO(n^d\log^{p-d-1}n\log\log n)$ which can then compute $|\phi^{1, p}_S(H_{1},\ldots,H_{p})|$ in $\bigO(\log^p n)$ expected time for any $p$-tuple of halfspaces.
\end{theorem}
\begin{proof}
The proof is by induction. Let us consider the base case $p=1$. We use the datastructure from \cref{thm:matousek_decomp} with $k'=1$ and $r=n/\log^{d/(d-1)}n$. To each inner set we add an attribute representing the size of the set. To each remainder set we attach the corresponding datastructure from \cref{thm:chan_composed_query} with $k'=1$. We can then compute $|\phi^{1}_S(H)|$ by using the primary datastructure to find the decomposition into inner sets and remaining set $S_\mathcal{R}$ given by \cref{thm:matousek_decomp}, and use the secondary attached datastructure to compute $|\phi^{1}_{S_\mathcal{R}}(H)|$. Because there are $\bigO(r^d)$ remainder sets, all of size at most $n/r$, the total preprocessing time is $\bigO(nr^{d-1} + r^d(n/r)\log(n/r)) \subset \bigO(n^d\log^{1-d-1}n\log\log n)$. Because in each query there is only one remainder set to consider, the total expected query time is $\bigO(\log r + (n/r)^{(d-1)/d}) \subset \bigO(\log n)$. Thus the claim holds for $p=1$.

Let us now suppose it holds for $1\leq p-1 < k$ and show it holds for $p$. We again use the datastructure from $\cref{thm:matousek_decomp}$ with $r=n/\log^{d/(d-1)}n$, this time with $k'=p$. To each inner set we attach the corresponding datastructure we inductively suppose exists for $p-1$. To each remainder set we attach the corresponding datastructure from \cref{thm:chan_composed_query} with $k'=p$. Using the knowledge about the distribution of sizes of the inner sets in \cref{thm:matousek_decomp}, the total preprocessing time will be
\begin{align*}
    &\bigO\left(nr^{d-1} + r^d(n/r)\log^p(n/r) + \sum_{i=1}^{\bigO(\log r)} \rho^i(\frac{n}{\rho^i})^d\log^{p-1-d-1} n \log\log n\right) \\
    &\subset \bigO\left(n^d\log^{-1}n + n^d\log^{p-d-2}n(\log\log n)^{p+1} + n^d\log^{p-d-2}\log\log n\sum_{i=1}^{\bigO(\log r)} (\rho^{1-d})^i \right)\\
    &\subset \bigO\left(n^d\log^{p-d-1}n\log\log n\right).
\end{align*}

The total expected query time will be
\begin{align*}
    &\bigO\left(\log r + \log r\log^{p-1}n + (n/r)^{(d-1)/d}\log^{p-1}(n/r) \right) \\
    &\subset \bigO\left(\log^p n \right).
\end{align*}

Thus by induction, the claim holds.
\end{proof}

For the final proof we need some other results from \cite{Chan2016}, which we summarize in the following theorem.


\begin{theorem}[\cite{Chan2016}]\label{thm:chan_tradeoff}
For any set of $n$ objects $S$ and any $1\leq B \leq n/\log^{\omega(1)}n$, we can build a datastructure in $\bigO(n\log^k n)$ time with the following properties.
\begin{itemize}
    \item There are $\bigO((n/B)\log^{k-1} n)$ subsets of $S$. We denote their collection $\mathcal{C}$. The subset of $\mathcal{C}$ consisting of sets of size at most $B$ is denoted as $\mathcal{B}$.
    \item For any $k$-tuple of halfspaces $H_{1},\ldots,H_{k}$, we can compute in $\bigO((n/B)^{(d-1)/d}\log^{k-1} n)$ expected time pointers to $t \in \bigO((n/B)^{(d-1)/d}\log^{k-1} n)$ sets of $\mathcal{C}$, denoted as $C_1,C_2\ldots C_{t}$ and another $t' \in \bigO((n/B)^{(d-1)/d}\log^{k-1} n)$ to sets of $\mathcal{B}$, denoted as $S_1,S_2\ldots S_{t'}$, all disjoint. Moreover they are such that
    
    \[\phi_S(H_{1},\ldots,H_{k}) = \left(C_1 \cup C_2 \cup \cdots \cup C_{t}\right) \cup \left(\phi_{S_1}(H_{1},\ldots,H_{k}) \cup \cdots \cup \phi_{S_{t'}}(H_{1},\ldots,H_{k})\right).\]
\end{itemize}
\end{theorem}

We are now ready for the final proof, which is an adaptation of the proof of Theorem 6.2 in \cite{Matousek1992}.
\begin{proof}[of \cref{thm:range-query}]
We use the datastructure from \cref{thm:chan_tradeoff}, for some unspecified $B$. To each set in $\mathcal{C}$, we add an attribute representing the size of the set. To each set in $\mathcal{B}$, we also attach the datastucture we get from \cref{thm:big_space_query} with $p=k$. Because each set in $\mathcal{B}$ is of size at most $B$ and there are  $\bigO((n/B)\log^{k-1} n)$ such sets, the total preprocessing time will be $\bigO(n\log^k n + m)$, where 
\begin{align*}
    m &:= (n/B)\log^{k-1}n + (n/B)(\log^{k-1}n) B^d(\log^{k-d-1}B)\log\log B \\
    &\in \bigO\left(B^{d-1}n(\log^{2k-d-2}n)\log\log n\right).
\end{align*}
We can make $m$ vary between $n\log^{k-1}n$ and $n^d/\log^{\omega(1)}n$ by having $B$ vary between $1$ and $n/\log^{\omega(1)}n$ (but we can still choose $m < n\log^{k-1}n$ in the statement of the theorem, as the $n\log^k n$ term then dominates the preprocessing). The total expected query time will be 
\begin{align*}
&\bigO\left((n/B)^{(d-1)/d} + (n/B)^{(d-1)/d}(\log^{k-1}n)\log^k B\right)\\
&\subset \bigO\left((n/B)^{(d-1)/d})\log^{k-1}n\log^{2k-1} n\right).
\end{align*} 
By rewriting this in terms of $m$ instead of $B$ we get an expected query time of
\[\bigO\left((n/m^{1/d})\log^{2(k+(k-d-1)/d))}n(\log\log n)^{1/d}\right).\]
\end{proof}

\section{Counting the number of pairwise intersecting intervals}
\label{app:intersect-segment}
Here we simply show how, when given $n$ intervals, we can compute the number of pairwise intersecting intervals in $\bigO(n\log(n))$ time.

Let $\F = \{[a_1,b_1],[a_2,b_2],\ldots,[a_n,b_n]\}$ be a set of intervals. Sort all the $a_i$'s and $b_i$'s (which we will call events from now on) in $\bigO(n\log(n))$ time. For simplicity, suppose all these events are distinct (otherwise, we could break ties in a manner that makes the following work). Now, go through the events in order while maintaining the number of `current' intervals $c$ and the total number of interval intersections $t$ thus far encountered:
\begin{itemize}
\item Each time an $a_i$ event is encountered, increase $t$ by $c$ before increasing $c$ by one.
\item Each time a $b_i$ event is encountered, decrease $c$ by one.
\end{itemize}
This is done in linear time.

It is not hard to see that every pair of intersecting intervals will be counted exactly once in $t$, namely when reaching the start event of the interval starting the latest in the pair. Thus, by the end of the execution, $t$ represents the quantity we are interested in. 

\bibliography{pq-arxiv}

\end{document}